\documentclass[journal,a4paper]{IEEEtran}
\usepackage{ifpdf}
\usepackage{cite}
\usepackage{graphicx}
\usepackage{amsmath,amsthm,amssymb}
\usepackage{algorithmic,algorithm}
\usepackage{array}
\usepackage{mdwmath}
\usepackage{mdwtab}
\usepackage{enumerate}
\usepackage{subfigure}
\theoremstyle{plain}

\newtheorem{thm}{Theorem}
\newtheorem{lm}{Lemma}

\newtheorem{pro}{Proposition}

\theoremstyle{definition}
\newtheorem{defi}{Definition}

\theoremstyle{remark}
\newtheorem*{rmk}{Remark}

\def\mb{\mathbf}
\def\bs{\boldsymbol}
\def\t{\text}

\begin{document}
\title{Asynchronous Distributed Downlink Beamforming and Power Control in Multi-cell Networks }
\author{ 
Siduo~Shen,~\IEEEmembership{Student Member,~IEEE,}
and~Tat-Ming~Lok,~\IEEEmembership{Senior Member,~IEEE,}
\thanks{S. Shen and T. M. Lok are with the Department of Information Engineering, The Chinese University of Hong Kong, Shatin, N.T., Hong Kong. E-mail: \{ssd010, tmlok\}@ie.cuhk.edu.hk.

This work has been submitted to the IEEE for possible publication.  Copyright may be transferred without notice, after which this version may no longer be accessible.
}
}
\maketitle
\begin{abstract}

In this paper, we consider a multi-cell network where every base station (BS) serves multiple users with an antenna array. 
Each user is associated with only one BS and has a single antenna. 
Assume that only long-term channel state information (CSI) is available in the system.
The objective is to minimize the network downlink transmission power needed to meet the users' signal-to-interference-plus-noise ratio (SINR) requirements.
For this objective, we propose an asynchronous distributed beamforming and power control algorithm which provides the same optimal solution as given by centralized algorithms.
To design the algorithm, the power minimization problem is formulated mathematically as a non-convex problem.
For distributed implementation, the non-convex problem is cast into the dual decomposition framework.
Resorting to the theory about matrix pencil, a novel asynchronous iterative method is proposed for solving the dual of the non-convex problem.
The methods for beamforming and power control are obtained by investigating the primal problem.
At last, simulation results are provided to demonstrate the convergence and performance of the algorithm.

\end{abstract}

\section{Introduction}

Downlink beamforming within a cell is an efficient technique for balancing the intra-cell interference, which has been studied since 1998 \cite{rashid-farrokhi_transmit_1998, visotsky_optimum_1999, bengtsson_optimal_1999,schubert_solution_2004,wiesel_linear_2006,yu_transmitter_2007}.
However, cell edge users may still suffer inter-cell interference in a multi-cell network.
By coordinating beamforming and power control among BSs, inter-cell interference can be mitigated \cite{bengtsson_optimal_2001}.
This kind of beamforming is called multi-cell beamforming.
Multi-cell beamforming has its applications, like interference management in femto-cell/small-cell networks, or spatial spectrum sharing in cognitive radio networks.

One of the basic beamforming problems is to minimize the total transmission power subject to each user's SINR constraint which is named the power minimization beamforming problem (PMBP). 
The problem differs for different kinds of CSI \cite{bengtsson_optimal_2001}. 
The PMBP is a convex problem \cite{wiesel_linear_2006,yu_transmitter_2007} with the instantaneous CSI which appears as a vector.
However, the PMBP is non-convex \cite{bengtsson_optimal_1999} with the long-term CSI. The long-term CSI appears as a spatial correlation matrix whose rank is greater than one.

Centralized algorithms in \cite{rashid-farrokhi_transmit_1998} and \cite{bengtsson_optimal_1999} are derived for single-cell PMBP with either instantaneous CSI or long-term CSI. Their generalizations are applied to the multi-cell PMBP and provide the optimal solution. However, their distributed implementations are unavailable.
The main reason is that they need strict coordination and synchronization among BSs which is impractical with nowadays backhaul network.

In this work, we assume the long-term CSI.
Our main contribution is an algorithm that can be implemented asynchronously in a distributed fashion. 
The algorithm achieves the same optimal solution as given by the centralized algorithms. 
In the algorithm, a BS updates the beamforming vectors of its users. 
A BS only communicates with its neighboring BSs for coordination via backhaul networks.
Motivated by the concept of ``power on demand'' \cite{lin2012distributed}, we let each user compute its required power and send it to its BS. The BSs allocate the users the required power.
The term ``asynchronous'' has two folds of meanings. First, neither the BSs nor the users need to update simultaneously, or with the same frequency. This property makes the algorithm suitable for distributed implementation.
Second, the algorithm converges with outdated coordination information. This implies that the algorithm is robust to the transmission delay and the packet loss of the backhaul network.

To design the algorithm, the PMBP is cast into the dual decomposition framework \cite{bertsekas1989parallel}. 
A novel asynchronous iterative method is proposed for solving the dual of the PMBP, which relies on the recent progress about the semidefiniteness of a matrix pencil \cite{stewart1979pertubation,kovavc1995trace,Liang20133085}. 
With the optimal dual variables, a user's beamforming vector is the null space of a linear combination of matrices which is computed by its BS.
The power control algorithm falls in the classical iterative power control framework in \cite{yates1995framework}. Thus, the users can update the power asynchronously with local information.
Besides, we also investigate the feasibility of the multi-cell PMBP with long-term CSI. The algorithm is derived assuming the PMBP being feasible. 

\subsection{Related Work}

In the literature about PMBP, uplink-downlink duality is a major aspect \cite{rashid-farrokhi_transmit_1998,visotsky_optimum_1999,schubert_solution_2004}.
The dual of the single-cell PMBP has a similar form to the uplink beamforming problem. 
Based on the dual of the PMBP, the virtual uplink problem is defined.
 The virtual uplink problem can be solved with a beamformer-power update method. 
The method is extended to the multi-cell case in \cite{botella2008coordination} by viewing multiple BSs as an antenna array. 
The communication requirement may be difficult to meet in practice.
Besides, method with limited coordination among BSs is proposed in \cite{huang_distributed_2011}. 

The duality of the single-cell PMBP with instantaneous CSI is also studied with the Lagrangian theory in \cite{wiesel_linear_2006},\cite{yu_transmitter_2007}.
An algorithm that is different from the algorithm in \cite{rashid-farrokhi_transmit_1998} is proposed. 
The work of \cite{dahrouj_coordinated_2010} is an extension of \cite{yu_transmitter_2007} to the multi-cell case. 
An efficient distributed multi-cell beamforming algorithm is proposed.
In \cite{dahrouj_coordinated_2010}, the system is time division duplex, i.e., reciprocal channel. Besides, instantaneous CSI is assumed, and the noise variances at each receiver are the same. In this case, the virtual uplink is the real uplink \cite{schubert_solution_2004}, which is convenient for distributed implementation.
However, the algorithm is not amenable to the multi-cell PMBP with long-term CSI.
Our work complements the work of \cite{dahrouj_coordinated_2010} by assuming long-term CSI, non-reciprocal channel, and different noise variances at each user.

Another class of methods is based on the convex optimization techniques \cite{bengtsson_optimal_1999, gershman_convex_2010}.
The PMBP is relaxed to a convex semidefinite program (SDP). In \cite{bengtsson_optimal_1999}, the relaxation gap is proved to be zero if the PMBP is feasible, which implies strong duality.
This conclusion is generalized to the multi-cell case in \cite{bengtsson_optimal_2001}.
With this approach, additional constraints can be considered in problem formulation \cite{vucic_robust_2009,huang_dual_2010}. However, this method is hard to be implemented distributively.

Other related works are summarized as follows.
In \cite{dartmann_duality_2013}, the surrogate duality approach is used to maximize the minimum SINR under power constraint.
In \cite{Xiang2013coordinated}, power minimization of multi-cell multi-cast beamforming system is studied.
In \cite{cai2012maxmin}, max-min weighted SINR problem subject to weighted sum power constraint is studied, where the problem is decouple into multiple sub-problems by introducing a set of slack variables.
In \cite{jeong2011beamforming}, \cite{noh_beamforming_2013},
beamforming techniques are applied to hierarchical systems for interference reduction.
Besides, there are works with other design objectives, such as transmission rate \cite{zhang_cooperative_2010}, energy-efficiency \cite{lim_energy-efficient_2013}, and robustness\cite{wang_robust_2013}.

\subsection{Organization}
The rest of this paper is organized as follows. In Section \ref{sec:System Model and Problem Formulation}, the system model and the problem formulation are given. In Section \ref{sec:Problem Feasibility}, the feasibility of the PMBP is studied.
Then, in Section \ref{sec:Computational Solution for the MP}, we introduce the computation methods for the solution of the PMBP. In Section \ref{sec:Distributed Beamforming and Power Control}, we show how to implement the computation methods in an asynchronous distributed manner.
The performance is further demonstrated in Section \ref{sec:Simulation} via simulation.

\emph{Notation:} Let $\mb{A}$ be a matrix.
$\mb{A}^{-1}$, $\mb{A}^{H}$ and $\mb{A}^{T}$ represent the inverse of $\mb{A}$, the Hermitian transpose of $\mb{A}$ and the transpose of $\mb{A}$, respectively.
Besides, 
$\text{det}(\mb{A})$,
$\text{rank}(\mb{A})$,
$\text{null}(\mb{A})$, and 
$\rho(\mb{A})$
stand for the determinant of $\mb{A}$, 
the rank of $\mb{A}$,
the null space of $\mb{A}$,
and the radius of $\mb{A}$, respectively.
If $\mb{A}\succeq / \succ 0$, $\mb{A}$ is a positive semidefinite$/$positive definite matrix.
$\mb{I}$ denotes an identity matrix with adaptive size.
$\mathbb{E}[\cdot]$ denotes the statistical expectation. 
When describing vectors' relation, ``$\le$'', ``$\ge$'', ``$<$'', ``$>$'' are component-wise.

\section{System Model and Problem Formulation}
\label{sec:System Model and Problem Formulation}
We consider a multi-cell network with $M$ cells indexed by $m\in\mathcal{M} =\{1,...,M\}$. Each cell has a BS with $N$ antennas and $K$ users. Within a cell, a user is indexed by $i\in\mathcal{I} = \{1,...,K\}$. In the system scope, a user is identified by its cell index and its intra-cell index, i.e., $(m,i)$. 
User $(m,i)$ indicates the user $i$ in cell $m$. For later use, ``$\forall (m,i)$'' represents ``for all users''.

The cells work in the same spectrum. 
The channel gain from BS $m$ to user $(n,i)$ is denoted with $\bs{h}_{m,n,i}\in\mathbb{C}^{N\times 1}$, $m,n\in\mathcal{M},~i\in\mathcal{I}$. 
We assume flat-fading channels, i.e., no inter-symbol-interference.
The long-term CSI is given as the spatial correlation matrices, 
$$
	\mb{R}_{m,n,i} = \mathbb{E}[\bs{h}_{m,n,i}\bs{h}_{m,n,i}^{H}],~\forall m,n\in\mathcal{M},~\forall i\in\mathcal{I},
$$
which satisfy $\text{rank}(\mb{R}_{m,n,i}) > 1$ and $\mb{R}_{m,n,i}\succeq 0$.
The long-term CSI assumption implies that the mean SINR is considered. 

Let $s_{m,i}\in\mathbb{C}$ denotes the information signal transmitted from BS $m$, $m\in\mathcal{M}$, to its user $i$, $i\in\mathcal{I}$.
Let $\bs{w}_{m,i}\in\mathbb{C}^{N\times 1}$ be the beamforming vector for $s_{m,i}$.
In the remaining of this work, we have $\|\bs{w}_{m,i}\| = 1$ as default.
The received signal of user $(m,i)$ is
\begin{align}
	y_{m,i} &= \bs{h}_{m,m,i}^{H}\bs{w}_{m,i}s_{m,i} 
	+ \sum_{j \neq i}\bs{h}_{m,m,i}^{H}\bs{w}_{m,j}s_{m,j} \notag \\
	&+ \sum_{n\neq m,~j}\bs{h}_{n,m,i}^{H}\bs{w}_{n,j}s_{n,j}
	+ z_{m,i}, \notag 
\end{align}
where $z_{m,i}$ is the additive white circularly symmetric Gaussian complex noise with variance $\sigma_{m,i}^{2}$. 
To simplify the notations, as \cite{dahrouj_coordinated_2010} does, we have $\sum_{j\neq i} = \sum_{\forall j\in\mathcal{I},j\neq i}$ and $\sum_{n\neq m,~j} = \sum_{\forall n\in\mathcal{M},n\neq m}\sum_{\forall j\in\mathcal{I}}$. 
Similar notation schemes are used in the rest of this work.

The SINR of the user $(m,i)$ is defined as follows,
\begin{align}
\textrm{sinr}_{m,i} = 
\frac{ 
p_{m,i}\bs{w}_{m,i}^{H}\mb{R}_{m,m,i}\bs{w}_{m,i}
}
{
\sum_{(n,j) \neq (m,i)}p_{n,j}\bs{w}_{n,j}^{H}\mb{R}_{n,m,i}\bs{w}_{n,j} 
+ \sigma_{m,i}^{2}
}, \notag
\end{align}
where $p_{m,i} = \mathbb{E}[\|s_{m,i}\|^{2}]$ is the allocated power of BS $m$ for user $(m,i)$ and $\bs{w}_{m,i}^{H}\mb{R}_{m,n,j}\bs{w}_{m,i} = \mathbb{E}[\|\bs{h}_{m,n,j}^{H}\bs{w}_{m,i}\|^{2}]$.
Besides, the notation $\sum_{(n,j) \neq (m,i)}$ represents summation over all the users except user $(m,i)$. 
In addition, we denote by $\gamma_{m,i}$ the SINR target of user $(m,i)$.

The objective is to minimize the total transmission power subject to each user's SINR constraint.
Mathematically, the objective is to find the solution of the following \emph{master problem},
\begin{align}
\label{master problem}
\underset{\bs{p}\ge 0,\mb{\Upsilon}}{\text{min.}}~~\sum_{m,i} p_{m,i} 
~~~~\text{s.t.}~~\textrm{sinr}_{m,i} \ge \gamma_{m,i},~\forall (m,i),
\end{align}
where we define $\bs{p} = [p_{1,1},p_{1,2},...,p_{M,K}]^{T}$ and $\mb{\Upsilon} = [\bs{w}_{1,1},\bs{w}_{1,2},...,\bs{w}_{M,K}]$. 
Since the constraints, $\|\bs{w}_{m,i}\| = 1,~\forall (m,i),$ are default, we omit them in this and the following optimization problems.

\section{Problem Feasibility}
\label{sec:Problem Feasibility}

The master problem \eqref{master problem} is a mathematical formulation of the multi-cell PMBP. 
If the master problem is feasible, all users' SINR targets can be achieved.
However, the master problem is not feasible for all sets of SINR targets. 
In this section, we study the feasibility of the master problem which can be viewed as a multi-cell extension of the work in \cite{shendesign}.

The master problem is feasible if its feasible region is non-empty.
The feasible region is characterized by $\textrm{sinr}_{m,i} \ge \gamma_{m,i},~\forall (m,i)$, which can be written as,
\begin{align}
	\label{sinr equivalent form}
	&p_{m,i} - \gamma_{m,i} \sum_{(n,j)\neq(m,i)}p_{n,j}\frac{\bs{w}_{n,j}^{H}\mb{R}_{n,m,i}\bs{w}_{n,j} }{ \bs{w}_{m,i}^{H}\mb{R}_{m,m,i}\bs{w}_{m,i} } \notag \\
	\ge &\frac{\gamma_{m,i}\sigma_{m,i}^{2}}{\bs{w}_{m,i}^{H}\mb{R}_{m,m,i}\bs{w}_{m,i} },~\forall (m,i).
\end{align}
The norm constraints of the beamforming vectors can be guaranteed by normalization, so they do not impact the feasibility of the master problem.

We rewrite the above inequalities in the following compact form,
\begin{align}
\label{feasible region in vector}
	(\mb{I} - \mb{\Gamma}\mb{G})\bs{p}\ge \bs{\eta},
\end{align}
where $\mb{\Gamma}$ is a diagonal matrix, $\bs{\eta}$ is a vector, and $\mb{G}$ is a $MK \times MK$ matrix.
In particular, we have 
$$
\mb{\Gamma} = \text{diag}(\gamma_{1,1},\gamma_{1,2},...,\gamma_{M,K}),
$$ 
the component of $\bs{\eta}$ as
\begin{align}
	\label{define of eta}
	\eta_{m,i} = \frac{\gamma_{m,i}\sigma_{m,i}^{2}}{\bs{w}_{m,i}^{H}\mb{R}_{m,m,i}\bs{w}_{m,i} },
\end{align}
and the component of matrix $\mb{G}$ as
\begin{align}
\label{eq:matrix G}
[\mb{G}]_{(m,i),(n,j)} = \frac{ \bs{w}_{n,j}^{H}\mb{R}_{n,m,i}\bs{w}_{n,j} }
{ \bs{w}_{m,i}^{H}\mb{R}_{m,m,i}\bs{w}_{m,i} }.
\end{align}

The feasibility of the master problem now depends on the inequality \eqref{feasible region in vector}.
Mathematically, according to the Perron-Frobenius theorem \cite{horn1990matrix}, we have the following lemma.
\begin{lm}
\label{feasibility condition}
The master problem is feasible if and only if there exists a set of beamforming vectors such that $\rho(\mb{\Gamma}\mb{G} ) <1$.
\end{lm}

Similar conclusion can be found in \cite{bengtsson_optimal_2001,schubert_solution_2004}.
From system perspective, given the users' SINR targets, $\mb{\Gamma}\mb{G}$ captures the interference relationship among the users which mainly depends on the physical environment.
If $\rho(\mb{\Gamma}\mb{G} ) <1$, the lower bound of $\bs{p}$ would be $(\mb{I} - \mb{\Gamma}\mb{G})^{-1}\bs{\eta}$. According to \eqref{define of eta}, $\bs{p}$ is in connection with the thermal noise.
Since the feasibility of the master problem depends on $\mb{\Gamma}\mb{G}$, to further reveal the physical meaning of $\rho(\mb{\Gamma}\mb{G})< 1$, motivated by \cite{schubert_solution_2004,jong_2010_design}, we have the following proposition.
\begin{pro}
\label{pro:feasibility}
The master problem is feasible if and only if there exists a set of beamforming vectors such that the following inequalities hold for any $\bs{p}\ge 0$,  
\begin{align}
\label{feasibility conditions 2}
  \frac{p_{m,i}\bs{w}_{m,i}^{H}\mb{R}_{m,m,i}\bs{w}_{m,i}}
	{\sum_{(n,j)\neq (m,i)}p_{n,j}\bs{w}_{n,j}^{H}\mb{R}_{n,m,i}\bs{w}_{n,j}
	} > \gamma_{m,i},~\forall (m,i).
\end{align}
\end{pro}
\begin{proof}
See Appendix \ref{proof of the feasibility pro}.
\end{proof}

The ratio on the left hand side of inequality \eqref{feasibility conditions 2} implies that the signal to interference ratio should be strictly larger than the target SINR. It shows mathematically that the feasibility mainly depends on the existence of the beamforming vectors that can coordinate the interference.

In summary, Lemma 1 and Proposition 1 can be viewed as further characterizations of the feasible region of the master problem.
They can provide practical design insights.
The above results remind us that feasibility can be achieved by removing the strongly interfered users, or decreasing some users' SINR targets. As shown in \cite{somekh_cth11-2_2006,vemula2006inter,stridh_system_2006}, user scheduling and network-wide congestion control can be employed for feasibility.
To focus on the algorithm design, we assume the master problem is feasible in the derivation of the algorithm. 

\section{Computational Solutions}
\label{sec:Computational Solution for the MP}

In this section, we derive a series of computation methods for solving the mathematical problem \eqref{master problem}. They are the underlying principles of the asynchronous distributed beamforming and power control algorithm.
We derive these methods with their asynchronous distributed implementations in mind. 

For distributed implementation, we start by employing the Lagrangian dual decomposition method \cite{bertsekas1989parallel,wiesel_linear_2006,yu_transmitter_2007}. Furthermore, since strong duality holds for the master problem, which is proved in \cite{bengtsson_optimal_2001}, the optimality of the dual decomposition method is guaranteed.
The Lagrangian of the master problem is written as follows,
\begin{align}
\label{eq:lagrangian0}
	&L(\bs{p},\mb{\Upsilon}, \bs{\lambda}) 
= \sum_{m,i}\lambda_{m,i}\sigma_{m,i}^{2} + \sum_{m,i}p_{m,i}\bs{w}_{m,i}^{H}\bs{w}_{m,i} \notag \\
&- \sum_{m,i}\lambda_{m,i}\left(1+\frac{1}{\gamma_{m,i}}\right)p_{m,i}\bs{w}_{m,i}^{H}\mb{R}_{m,m,i}\bs{w}_{m,i}
 \notag \\
&+ \sum_{m,i}\lambda_{m,i}\sum_{n,j}p_{n,j}\bs{w}_{n,j}^{H}\mb{R}_{n,m,i}\bs{w}_{n,j},
\end{align}
where $\bs{\lambda} = [\lambda_{1,1},\lambda_{1,2},...,\lambda_{M,K}]^{T}$ and $\bs{\lambda}\ge 0$.
About the last term of the above equation, we have  
\begin{align}
	&\sum_{m,i}\lambda_{m,i}\sum_{n,j}p_{n,j}\bs{w}_{n,j}^{H}\mb{R}_{n,m,i}\bs{w}_{n,j} \notag \\
= &\sum_{m,i}\sum_{n,j}p_{m,i}\bs{w}_{m,i}^{H}\lambda_{n,j}\mb{R}_{m,n,j}\bs{w}_{m,i}. \notag
\end{align}

Then, equation \eqref{eq:lagrangian0} can be rewritten as,
\begin{align}
\label{eq:lagrangian2}
	&L(\bs{p},\mb{\Upsilon}, \bs{\lambda}) 
= \sum_{m}\sum_{i}\lambda_{m,i}\sigma_{m,i}^{2} \notag \\
&+ \sum_{m}\Bigg(
\sum_{i}p_{m,i}\bs{w}_{m,i}^{H}
\Bigg( 
\mb{I} - \left(1+\frac{1}{\gamma_{m,i}}\right)\lambda_{m,i}\mb{R}_{m,m,i}
\notag \\	
&+ \sum_{n,j}\lambda_{n,j}\mb{R}_{m,n,j}
\Bigg)\bs{w}_{m,i} \Bigg).
\end{align}

Furthermore, we define 
\begin{align}
	\label{eq:sub lagrangian}
	&L_{m}(\bs{p}_{m},\mb{\Upsilon}_{m}, \bs{\lambda}) = \sum_{i}\lambda_{m,i}\sigma_{m,i}^{2} \notag \\
&+\sum_{i}p_{m,i}\bs{w}_{m,i}^{H}
\Bigg( 
\mb{I} - \left(1+\frac{1}{\gamma_{m,i}}\right)\lambda_{m,i}\mb{R}_{m,m,i}
\notag \\
	&+ \sum_{n,j}\lambda_{n,j}\mb{R}_{m,n,j}
\Bigg)\bs{w}_{m,i},
\end{align}
where $\bs{p}_{m} = [p_{m,1},...,p_{m,K}]^{T}$ and $\mb{\Upsilon}_{m} = [\bs{w}_{m,1},...,\bs{w}_{m,K}]$.

Then, the Lagrangian is 
\begin{align}
	\label{eq: lagrangian decomposition}
	L(\bs{p},\mb{\Upsilon}, \bs{\lambda}) = \sum_{m}L_{m}(\bs{p}_{m},\mb{\Upsilon}_{m}, \bs{\lambda}).
\end{align}

According to the Lagrangian theory, the \emph{primal problem} is
\begin{align}
\label{PP}
\underset{\bs{p}\ge 0,\mb{\Upsilon}}{\text{min.}}~~L(\bs{p},\mb{\Upsilon}, \bs{\lambda}).
\end{align}

Based on equation \eqref{eq: lagrangian decomposition}, the primal problem can be decomposed into $M$ sub-primal problems,
\begin{align}
\label{sub-PP m}
~~~\underset{\bs{p}_{m}\ge 0,\mb{\Upsilon}_{m}}{\text{min.}}
~~L_{m}(\bs{p}_{m},\mb{\Upsilon}_{m}, \bs{\lambda}).
\end{align}
With the optimal dual variables, the beamforming and power control can be performed within each BS by solving the corresponding sub-primal problem.

To avoid being trivial, we require, $\forall (m,i)$,
\begin{align}
	\label{semidefinte condition}
	\mb{I} - \left(1+\frac{1}{\gamma_{m,i}}\right)\lambda_{m,i}\mb{R}_{m,m,i}
	+ \sum_{n,j}\lambda_{n,j}\mb{R}_{m,n,j} \succeq 0.
\end{align}
Otherwise, the optimum of the primal problem would be negative infinity.

From \eqref{PP} and \eqref{semidefinte condition}, the \emph{dual problem} is
\begin{align}
	\label{DP}
\underset{\bs{\lambda}\ge 0}{\text{max.}}
~~&\sum_{m,i}\lambda_{m,i}\sigma_{m,i}^{2} \notag \\
\text{s.t.}~~&\mb{I} - \left(1+\frac{1}{\gamma_{m,i}}\right)\lambda_{m,i}\mb{R}_{m,m,i}
	+ \sum_{n,j}\lambda_{n,j}\mb{R}_{m,n,j} \succeq 0, \notag \\ 
	&~~~~~~~~~~~~~~~~~~~~~~~~~~~~~~~~~~~~~~~~~~~~~~~\forall (m,i).	
\end{align}

The dual problem can be centrally solved by CVX \cite{grant2008cvx}.
However, a computation method designed for distributed implementation is still unavailable. 

If instantaneous CSI is used, inequalities \eqref{semidefinte condition} are linear inequalities in essence. 
When the corresponding dual problem achieves optimality, they are all active, i.e., establishing with equality.
Based on this property, an efficient iterative method is derived in \cite{wiesel_linear_2006,yu_transmitter_2007,dahrouj_coordinated_2010}.

However, since the ranks of the spatial correlation matrices are greater than one, the inequalities in \eqref{semidefinte condition} are linear matrix inequalities. 
This difference makes it challenging to derive a similar iterative algorithm.

\subsection{Methods for Computing Dual Variables}

One of our major contributions is an asynchronous iterative method for computing the dual variables. To design the method, we first propose the following lemma.
\begin{lm}
\label{minimum non-negative eigenvalue}
Given $N\times N$ Hermitian matrices, $\mb{A}\succ 0$, $\mb{B}\succeq 0 $, 
the minimum non-negative eigenvalue of matrix pencil $\mb{A} - \mu\mb{B}$,  denoted by $\mu_{+}(\mb{A},\mb{B})$, exists and is finite.
Assuming variable $\mu\ge 0$, we have
$$
	\mb{A} - \mu\mb{B} \succeq 0  ~\Leftrightarrow~  \mu \le \mu_{+}(\mb{A},\mb{B}).
$$
\end{lm}
\begin{proof}
	See Appendix \ref{proof of min-eig theorem}.
\end{proof}
\begin{rmk}
The related definitions about matrix pencil are provided in Appendix \ref{proof of min-eig theorem}. 
The general method to compute the eigenvalues of a matrix pencil is based on the generalized Schur decomposition \cite{golub1996matrix} (see Definition \ref{gener schur} in Appendix \ref{proof of min-eig theorem}).
When matrix $\mb{B}$ is non-singular, the set of eigenvalues of matrix pencil $\mb{A} - \mu\mb{B}$ is the same as the set of eigenvalues of matrix $\mb{A}\mb{B}^{-1}$.
\end{rmk}

Next, we rewrite the constraints of the dual problem \eqref{DP} as 
\begin{align}
	\label{constraint for lambda}
	&\left(1+\frac{1}{\gamma_{m,i}}\right)^{-1}
	\left(
	\mb{I}+ \sum_{n,j}\lambda_{n,j}\mb{R}_{m,n,j}
	\right)
	 - \lambda_{m,i}\mb{R}_{m,m,i}
	 \succeq 0, \notag \\ 
	&~~~~~~~~~~~~~~~~~~~~~~~~~~~~~~~~~~~~~~~~~~~~~~~~~~~~~~\forall (m,i).
\end{align}
Since correlation matrices $\mb{R}_{n,m,i}\succeq 0,~\forall n,m\in\mathcal{M},~\forall i\in\mathcal{I}$, we have $\left(
	\mb{I}+ \sum_{n,j}\lambda_{n,j}\mb{R}_{m,n,j}
	\right) \succ 0$ and $\mb{R}_{m,m,i} \succeq 0$. According to Lemma \ref{minimum non-negative eigenvalue}, inequalities in \eqref{constraint for lambda} are equivalent to
	\begin{align}
	\label{apply theorem}
		&\lambda_{m,i} \le \notag \\
	  &\mu_{+}\Bigg( \left(1+\frac{1}{\gamma_{m,i}}\right)^{-1}
	\left(
	\mb{I}+ \sum_{n,j}\lambda_{n,j}\mb{R}_{m,n,j}
	\right) , \mb{R}_{m,m,i} \Bigg),\notag \\ 
	&~~~~~~~~~~~~~~~~~~~~~~~~~~~~~~~~~~~~~~~~~~~~~~~~~~~~~~\forall (m,i).
	\end{align}


Furthermore, since the right hand side of \eqref{apply theorem} also contains the dual variables, we define the following functions,
\begin{align}
	\label{eq:min-eig mapping}
	&J_{m,i}\left( \bs{\lambda} \right) = \notag \\
	&\mu_{+}\left( \left(1+\frac{1}{\gamma_{m,i}}\right)^{-1}
	\left(
	\mb{I}+ \sum_{n,j}\lambda_{n,j}\mb{R}_{m,n,j}
	\right) , \mb{R}_{m,m,i} \right), \notag \\	&~~~~~~~~~~~~~~~~~~~~~~~~~~~~~~~~~~~~~~~~~~~~~~~~~~~~~~\forall (m,i).
	\end{align}

Combining \eqref{apply theorem} and \eqref{eq:min-eig mapping}, we have
$		\bs{\lambda} \le \bs{J}\left( \bs{\lambda} \right), $
where $\bs{J}(\bs{\lambda}) = \left[J_{1,1}(\bs{\lambda}),J_{1,2}(\bs{\lambda}),...,J_{M,K}(\bs{\lambda})\right]^{T}$.
Thus far, the dual problem \eqref{DP} is equivalent to
$$
\underset{\bs{\lambda}\ge 0}{\text{max.}}
~~\bs{\sigma}^{T}\bs{\lambda}~~~~~
\text{s.t.}
~~\bs{\lambda} \le \bs{J}\left( \bs{\lambda} \right),
$$
where $\bs{\sigma} = [ \sigma_{1,1}^{2},\sigma_{1,2}^{2},...,\sigma_{M,K}^{2} ]^{T}$.

To solve the above problem, we propose an \emph{iterative dual computation method} which converges under asynchronous implementation. Briefly, the method is to update the dual variables with $\bs{J}\left( \bs{\lambda} \right)$.
We first propose the following theorem about its synchronous implementation.
\begin{thm}
\label{dual convergence}
From any feasible initial point $\bs{\lambda}(0)$,
the sequence generated by the following iterations,
$$
\bs{\lambda}(t+1) = \bs{J}\left( \bs{\lambda}(t) \right),~t = 0,1,2,...,
$$
converges to a unique point $\bs{\lambda}^{\star}$ which satisfies
	$
	\bs{\lambda}^{\star} = \bs{J}\left( \bs{\lambda}^{\star} \right),
	$
and is the optimal solution of the dual problem.
\end{thm}
\begin{proof}
See Appendix \ref{proof of the dual convergence}.
\end{proof}

\begin{rmk}
The key point of proving the above theorem is to prove that the specific function $\bs{J}(\bs{\lambda})$ is a standard function as defined in \cite{yates1995framework}. Besides, the work in \cite{cai2011unified} proved the convergence of a similar fixed point iterations for instantaneous CSI.
Moreover, the key feature of the synchronous implementation is that when updating at the $t$th iteration, all the dual variables are updated with the same input dual variable $\bs{\lambda}(t)$.
When the dual variables are updated with the different input dual variables, the iterative dual computation method is implemented asynchronously. 
\end{rmk}

We employ the asynchronous algorithm model in \cite{bertsekas1989parallel}. Define $\bs{\lambda}( \bs{f}_{m,i}(t) )$ as the input dual variables used for generating $\lambda_{m,i}(t+1)$, where 
\begin{align}
\label{asyn time index}
	\bs{\lambda}( \bs{f}_{m,i}(t) ) &= [\lambda_{1,1}(f_{m,i}^{1,1}(t)),\lambda_{1,2}(f_{m,i}^{1,2}(t)), \notag \\ 
	&~~~~...,\lambda_{M,K}(f_{m,i}^{M,K}(t)) ]^{T}. \\
	\bs{f}_{m,i}(t) &= [f_{m,i}^{1,1}(t),f_{m,i}^{1,2}(t),...,f_{m,i}^{M,K}(t)]^{T}.
\end{align}
Function $f_{m,i}^{n,j}(t) \in \{0,1,...,t-1,t\}$ is the time index of dual variable $\lambda_{n,j}$ used for updating $\lambda_{m,j}(t+1)$.
So, $\lambda_{n,j}(f_{m,i}^{n,j}(t))$ is the value of the dual variable $\lambda_{n,j}$ used for updating $\lambda_{m,i}(t+1)$.

The asynchronous model also captures the case that the dual variables are not updated in every iteration. Denote by $\mathcal{T}_{m,i}$ the set of time indexes at which dual variable $\lambda_{m,i}$ is updated. For example, $\mathcal{T}_{m,i} = \{1,3,7,...\}$.
We assume that $\mathcal{T}_{m,i}$ is infinite, and, provided any time $t_{0}$, there is $t_{1}$ such that $f_{m,i}^{n,j}(t) \ge t_{0}$, for all $t_{1}$. 
The physical meaning of this assumption is that the transmission of the dual variables does not fail in all the iterations.

The convergence of the iterative dual computation method under asynchronous implementation is formally given as the following theorem.
\begin{thm}
\label{asyn iter dual comp methd convergence}
From any feasible initial point $\bs{\lambda}(0)$,
the sequence generated by the following iterations, $\forall (m,i),$
$$
\lambda_{m,i}(t+1) = 
\left\{
\begin{array}{ll}
J_{m,i}\left( \bs{\lambda}( \bs{f}_{m,i}(t) ) \right), & t\in\mathcal{T}_{m,i}, \\
\lambda_{m,i}(t),& \t{otherwise},
\end{array}
\right.
$$
converges to a unique point $\bs{\lambda}^{\star}$ which satisfies
	$
	\bs{\lambda}^{\star} = \bs{J}\left( \bs{\lambda}^{\star} \right),
	$
and is the optimal solution of the dual problem.
\end{thm}
\begin{proof}
We have proved that $\bs{J}\left( \bs{\lambda} \right)$ is a standard function in Appendix \ref{proof of the dual convergence}. The proof of this theorem is also based on this property.
With the reasoning similar to \cite{yates1995framework}, this theorem can be proved based on Theorem \ref{dual convergence} and the conclusions from \cite[pp. 431-433]{bertsekas1989parallel}.
\end{proof}
In summary, the iterative dual computation method converges to the same optimal point $\bs{\lambda}^{\star}$ under both synchronous and asynchronous implementations.

\subsection{Methods for Computing Primal Variables}

We derive the methods for computing the primal variables in this section.

\subsubsection{Beamforming}
Thanks to the dual decomposition, each cell's beamforming problem realtes to one sub-primal problem.
The optimality conditions for the $m$th cell's sub-primal problem are
\begin{align}
\label{optimality condition for SP m}
\frac{ \partial L_{m}(\bs{p}_{m},\mb{\Upsilon}_{m}, \bs{\lambda}) }{ \partial \bs{w}_{m,i}}
= &\Bigg( 
\mb{I} - \left(1+\frac{1}{\gamma_{m,i}}\right)\lambda_{m,i}\mb{R}_{m,m,i}
	\notag \\
	&+ \sum_{n,j}\lambda_{n,j}\mb{R}_{m,n,j}
\Bigg)\bs{w}_{m,i} = \bs{0},~\forall i\in\mathcal{I}, 
\end{align}
which are equivalent to
\begin{align}
\label{bf equation}
	\bs{w}_{m,i} = \text{null}&\Bigg(
\mb{I} - \left(1+\frac{1}{\gamma_{m,i}}\right)\lambda_{m,i}\mb{R}_{m,m,i}
	\notag \\
	&+ \sum_{n,j}\lambda_{n,j}\mb{R}_{m,n,j}
	\Bigg),~\forall i\in\mathcal{I}. 
\end{align}

Furthermore, define 
$$
\mb{E}_{m,i}(\bs{\lambda}) = \mb{I} - \left(1+\frac{1}{\gamma_{m,i}}\right)\lambda_{m,i}\mb{R}_{m,m,i}
	+ \sum_{n,j}\lambda_{n,j}\mb{R}_{m,n,j}
$$
	which is Hermitian. 
If the matrix $\mb{E}_{m,i}(\bs{\lambda})$ is full rank, the null space will be a zero vector.
From the computation of the dual variables, we can see that the matrix $\mb{E}_{m,i}(\bs{\lambda}^{\star})$ has an eigenvalue to be $0$. Thus, $\mb{E}_{m,i}(\bs{\lambda}^{\star})$ is not full rank and has a non-zero null space.

It requires that the beamforming of each cell can only be performed after the dual variables are optimal.
The beamforming could be delayed for both numerical reasons, and practical reasons, like transmission delay.
Our method enables beamforming to be performed at any time. 

When $\mb{E}_{m,i}(\bs{\lambda})$ is full rank, we have the eigendecomposition of it as $\mb{V}_{m,i}\mb{\Lambda}_{m,i}\mb{V}_{m,i}^{-1}$, where $\mb{\Lambda}_{m,i}$ is a diagonal matrix and $\mb{V}_{m,i}$ is a unitary matrix. 
Let 
$$
k^{\prime} = \underset{k = 1,...,N}{\t{arg min~}}\|[\mb{\Lambda}_{m,i}]_{k,k}\|.
$$ 
Define a diagonal matrix $\widetilde{\mb{\Lambda}}_{m,i}$ whose components are the same as $\mb{\Lambda}_{m,i}$ except $[\widetilde{\mb{\Lambda}}_{m,i}]_{k^{\prime},k^{\prime}} = 0$.
The beamforming vector can be computed as,
\begin{align}
\label{method for BF}
	\bs{w}_{m,i} = \t{null}\left( \mb{V}_{m,i}\widetilde{\mb{\Lambda}}_{m,i}\mb{V}_{m,i}^{-1} \right).
\end{align}
\begin{rmk}
This method provides a set of suboptimal beamforming vectors when $\bs{\lambda}$ is not optimal.
Most importantly, this method does not affect the result when the dual variables are optimal. 
\end{rmk}
\subsubsection{Power Control}
Going back to the master problem with the optimal beamforming vectors, the master problem reduces to a typical power control problem \cite{yates1995framework},
\begin{align}
	\label{reduced MP}
	\underset{\bs{p}\ge 0}{\text{min.}}~~\sum_{m,i} p_{m,i}~~~
	\text{s.t.}~~p_{m,i} \ge I_{m,i}(\bs{p}),~\forall (m,i),
\end{align}
where 
\begin{align}
\label{power iteration function}
	&I_{m,i}(\bs{p}) \notag \\
 =&\left( \sum_{n,j}p_{n,j}g_{(n,j),(m,i)} + \sigma_{m,i}^{2}\right) \left( 1+ \frac{1}{\gamma_{m,i}}\right)^{-1} g_{(m,i),(m,i)}^{-1}, 
\notag \\
&~~~~~~~~~~~~~~~~~~~~~~~~~~~~~~~~~~~~~~~~~~~~~~~~~~~~\forall (m,i).
\end{align}
with power gains $g_{(n,j),(m,i)} = \bs{w}_{n,j}^{H}\mb{R}_{n,m,i}\bs{w}_{n,j},~\forall m,n\in\mathcal{M},\forall i\in\mathcal{I}$. 

Furthermore, let $\bs{I}(\bs{p}) = [I_{1,1}(\bs{p}),I_{1,2}(\bs{p}),...,I_{N,K}(\bs{p})]^{T}$.
Since $\bs{p} \ge \bs{I}(\bs{p})$,
the solution of problem \eqref{reduced MP} can be computed with the following iterative method,
\begin{align}
	\label{power iteration method}
	\bs{p}(t+1) = \bs{I}(\bs{p}(t)).
\end{align}
Moreover, this iterative power control method can be implemented asynchronously in a distributed fashion as given in \cite{yates1995framework}.

\section{Asynchronous Distributed Beamforming and Power Control}
\label{sec:Distributed Beamforming and Power Control}
\begin{figure}
	\centering
		\includegraphics[width=0.45\textwidth]{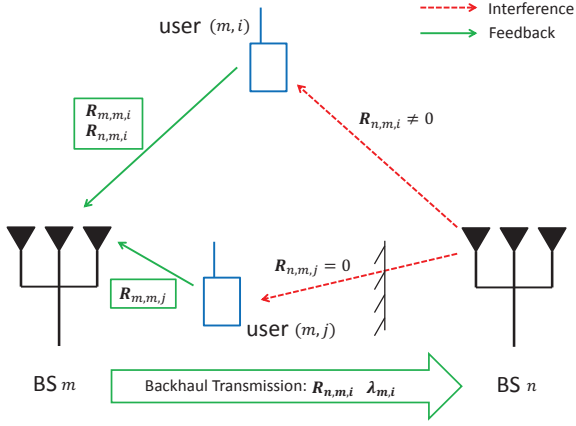}
	\caption{Illustration of information exchange among BSs}
	\label{fig:fig1}
\end{figure}

In this section, we introduce the asynchronous distributed beamforming and power control algorithm which is based on the above computation methods.
For ease of demonstrating the features of the algorithm, we describe the algorithm in the form of threads which are given in Table \ref{tab:TheDistributedBeamformingAndPowerControlAlgorithm}. The threads can work in parallel without strict coordination. Besides, with the proposed algorithm, the information exchange among BSs is illustrated in Fig. \ref{fig:fig1}.

\begin{table}[t]
	\centering
	\caption{The asynchronous distributed beamforming and power control algorithm}
	\rule{0.48\textwidth}{0.4pt}
	\begin{itemize}
		\item \textbf{CSI thread:}
		\begin{itemize}
			\item User $(m,i)$: \\
			estimate $\mb{R}_{m,m,i}$ and $\mb{R}_{n,m,i}$, if $\mb{R}_{n,m,i} \neq 0$;\\ 
			send them to BS $m$.
			\item BS $m$: \\
			receive $\mb{R}_{m,m,i}$ and $\mb{R}_{n,m,i}$, $\forall i\in\mathcal{I}$, from user $(m,i)$; \\
			send $\mb{R}_{n,m,i}$ ($\mb{R}_{n,m,i} \neq 0$) to BS $n$; receive $\mb{R}_{m,n,j}$ from BS $n$.
			\item The period of updating CSI: $T_{\t{CSI}}$.
		\end{itemize}
		\item \textbf{Dual computation (DC) thread:}
		\begin{itemize}
			\item BS $m$:\\
			receive $\lambda_{n,j}(t)$ from BS $n$ when $\mb{R}_{m,n,j}\neq 0$; \\
			update dual variable $\lambda_{m,i}(t+1)$ with $J_{m,i}\left( \bs{\lambda}( \bs{f}_{m,i}(t) ) \right)$, $\forall i\in\mathcal{I}$; \\
			send $\lambda_{m,i}(t+1)$ to BS $n$, if $\mb{R}_{n,m,i}\neq 0$ is reported by user $(m,i)$.
			\item The period between two iterations: $T_{\t{DC}}$.
		\end{itemize}
		\item \textbf{Beamforming (BF) thread:}
		\begin{itemize}
			\item BS $m$: \\
			fetch the CSI from the CSI thread and the dual variables from the DC thread; \\
			update $\bs{w}_{m,i}$ with equation \eqref{method for BF}.
			\item The period of beamforming: $T_{\t{BF}}$.
		\end{itemize}
		\item \textbf{Power control (PC) thread:}
		\begin{itemize}
			\item User $(m,i)$: \\
			update power $p_{m,i}(t+1)$ with $I_{m,i}(\bs{p}(t))$; \\
			feedback the demanded power to BS $m$.
			\item BS $m$:\\
			transmit to user $(m,i)$ with required power.			
			\item The period of power control: $T_{\t{PC}}$.
		\end{itemize}
	\end{itemize}	
	\rule{0.48\textwidth}{0.4pt}
	
	\label{tab:TheDistributedBeamformingAndPowerControlAlgorithm}
\end{table}

The CSI thread is used to collect the long-term CSI. Each user can estimate the long-term CSI, from its associated BS to it, and, from the neighboring BSs to it. This can be done through pilot channels. The users send the long-term CSI to their BSs. 
When a BS interferes with a user of another cell, or its user is interfered with by other BSs, it exchanges the long-term CSI with the related neighboring BSs.
Comparing to the algorithms using instantaneous CSI, the communication overheads here are relatively low. The larger $T_{\t{CSI}}$ is, the smaller the time average communication overhead is.
Moreover, we assume $T_{\t{CSI}} \gg \max\{ T_{\t{DC}},T_{\t{BF}} ,T_{\t{PC}}\}$, i.e., the long-term CSI can be viewed as fixed for the other threads.

The DC thread is the key to the asynchronous implementation of the algorithm.
Each BS runs a DC thread which updates the dual variables of its users.
Mathematically, the update of a dual variable needs all the dual variables. Indeed, in addition to its own dual variables, a BS only needs to know the dual variables of the users interfered by it. The BSs communicate via the backhaul network.
In practice, transmission delay and packet loss could happen, which means that some BSs use outdated information.
According to Theorem \ref{asyn iter dual comp methd convergence}, the BSs can still compute the optimal dual variables in this situation.
Furthermore, the BSs do not need to update the dual variables simultaneously,
nor with the same frequency.
As to the computational complexity, each BS needs to compute the generalized eigenvalues. 
If using the QZ method \cite{golub1996matrix}, the computational complexity of a BS per iteration is about $O(KN^{3})$.

Each BS runs a BF thread to compute the beamforming vectors of its users. 
With the beamforming computation method in \eqref{method for BF}, the update frequency of the BF thread does not need to coordinate with the DC thread. 
The needed dual variables in the computation are fetched from the DC thread.
The main computation load of the BF thread is to compute the eigenvalues and the null spaces of $K$ matrices with size $N\times N$.

The PC thread runs at each user. 
Each user computes the required power with equation \eqref{power iteration function} based on the local information:
the total received signal power and the power gain between it and its BS.
Each user sends the demanded power to its BS and the BS allocates the user the required power.
This kind of power control can be understood as ``power on demand'' \cite{lin2012distributed}.
Besides, this iterative power control method \eqref{power iteration function} falls in the iterative power control framework in \cite{yates1995framework}. Thus, the users can update the power asynchronously and distributively.
Besides, the computational complexity for each user is low. 

In summary, the DC thread is the key to the asynchronous implementation of the algorithm. As long as the optimal dual variables are achieved, the optimal beamforming and power control can be realized by the BSs and the users asynchronously.
In other words, if the iterative dual computation method could not converge under asynchronous implementation, the beamforming and power control algorithm could not converge asynchronously either.
With Theorem \ref{asyn iter dual comp methd convergence}, we conclude that the asynchronous distributed beamforming and power control algorithm can provide the optimal solution to the master problem.

\section{Simulation}
\label{sec:Simulation}

In this section, we demonstrate the performance of the proposed beamforming and power control algorithm via simulation.
The downlink channel model captures the path loss and the spatial correlation. Let $\mb{R}_{n,m,i} = \left(d_{n,m,i}\right)^{\chi}\bar{\mb{R}}_{n,m,i}$, $\forall m,n\in\mathcal{M},~\forall i\in\mathcal{I}$, where $d_{n,m,i}$ is the distance from BS $n$ to user $(m,i)$ and $\chi = - 3$.
According to \cite{bengtsson_optimal_1999}, the spatial correlation matrix is approximated by, $\forall m,n\in\mathcal{M},~\forall i\in\mathcal{I}$,
\begin{align}
	&[\bar{R}_{n,m,i}(\theta_{n,m,i},\sigma_{\theta})]_{k,l} \notag \\
	= & \text{e}^{j\pi(k-l)\sin\theta_{i}}\text{e}^{-\frac{\pi(k-l)\sigma_{\theta}\cos\theta_{n,m,i}}{2}},~k,l = 1,...,N, \notag
\end{align}
where $\theta_{n,m,i}$ means that user $(m,i)$ locates at $\theta_{n,m,i}$ relative to the array broadside of BS $n$.
Each user is surrounded by local scatters corresponding to a spread angle of $\sigma_{\theta} = 2^{\circ}$. 
Without loss of generality, the SINR targets of the users are assumed to be the same.
The other simulation settings are given as follows, if not specified.
Let $N = 4,~\sigma_{m,i}^{2} = 10^{-12},~\gamma_{m,i} = 0.1,~\forall (m,i)$.
The distance between two BSs is $2000$. The users are randomly located within their cells.


As pointed above, the iterative dual computation method is the key to the asynchronous implementation of the whole algorithm.
The asynchronous implementation here is to let each dual variable has a positive probability being not updated. The probability is denoted by $P$. 
This setup simulates the cases that the BSs do not update simultaneously and that packet loss or delay happens.


\begin{figure}[t]
\centering
	\subfigure[Asynchronous Implementation]{
	\centering
		\includegraphics[width=0.45\textwidth]{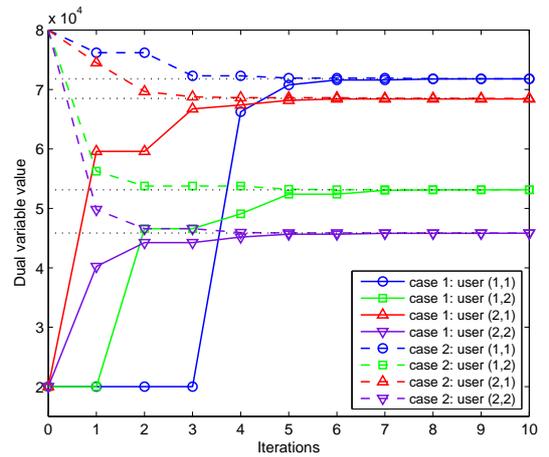}
		\label{fig:fig3}
	}
	\hspace{0.01\linewidth}
	\subfigure[Synchronous Implementation]{
		\centering
		\includegraphics[width=0.45\textwidth]{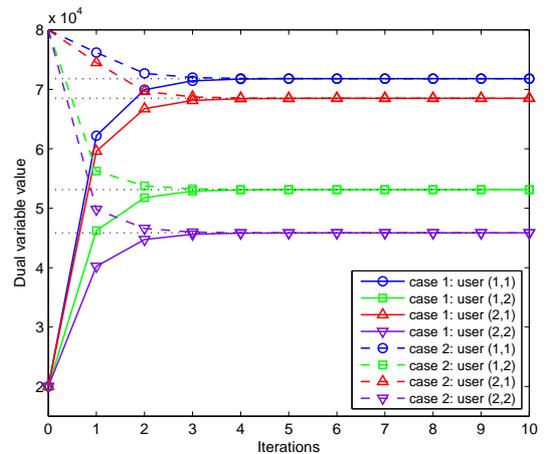}
		\label{fig:fig2}
	}
	\caption{The convergence of the iterative dual computation method under synchronous and asynchronous (P =0.5) implementations. Let $M =2,~K=2$. In case 1, the initial dual variables are 40; In case 2, the initial dual variables are 25. The black dot curves are the optimal solutions of the dual problem computed with CVX serving as benchmarks.}
\end{figure}

\begin{figure}[t]
		\centering
			\includegraphics[width=0.45\textwidth]{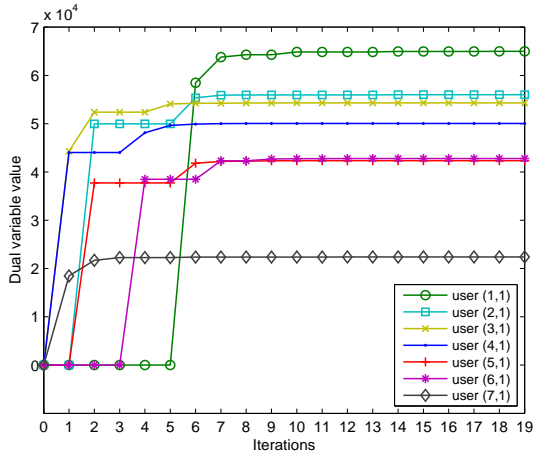}
		\caption{The convergence of the iterative dual computation method under asynchronous implementation. $P = 0.5,~M = 7,~K=2$. The seven cells are located according to the hexagonal model similar to \cite[Fig. 2]{cai2012maxmin}.}
		\label{fig:fig4}
\end{figure}

Fig. \ref{fig:fig3} shows the convergence of the iterative dual computation method under asynchronous implementation with $P=0.5$.
In case 1, we set $\bs{\lambda}(0)\le \bs{J}(\bs{\lambda}(0))$ resulting in increasing sequences. In case 2, we set $\bs{\lambda}(0)\ge \bs{J}(\bs{\lambda}(0))$ resulting in decreasing sequences.
To show the effect of the asynchronous implementation, the convergence of the iterative dual computation method under synchronous implementation is illustrated in Fig. \ref{fig:fig2}.  
Comparing Fig. \ref{fig:fig3} and Fig. \ref{fig:fig2}, 
we notice that the method under asynchronous implementation has a slower convergence speed. 
Besides, Fig. \ref{fig:fig4} shows the convergence of the iterative computation method with $M=7,~K=2,~P =0.5,$ where only the dual variables of each cell's first user are shown.

%

To further confirm this observation, we study the influence of the asynchronous implementation with $M=4,~K=4$. The four BSs are put at the four corners of a square.
Table \ref{tab:iterations} shows the average iteration numbers needed for the iterative dual computation method to converge under different values of $P$.
The convergence criteria is $\|\bs{\lambda}(t+1) - \bs{\lambda}(t)\| \le 10^{-5}$.
We can see that, as $P$ increases, the average number of iterations also increases. 
\begin{table}[t]
	\centering
		\begin{tabular}{|c|c|c|c|c|c|c|}
		\hline
		$P$ & 0 & 0.1 & 0.2 & 0.3 & 0.4 & 0.5 \\ 
		\hline
		Num. of Iter.	& 24 & 28 & 33 & 39 & 46 & 56 \\  
		\hline
		\end{tabular}
	\caption{Average numbers of convergence iterations versus failure probability $P$. $M = 4,~K=4$}
	\label{tab:iterations}
\end{table}


\begin{figure}[t]
	\centering
		\includegraphics[width=0.45\textwidth]{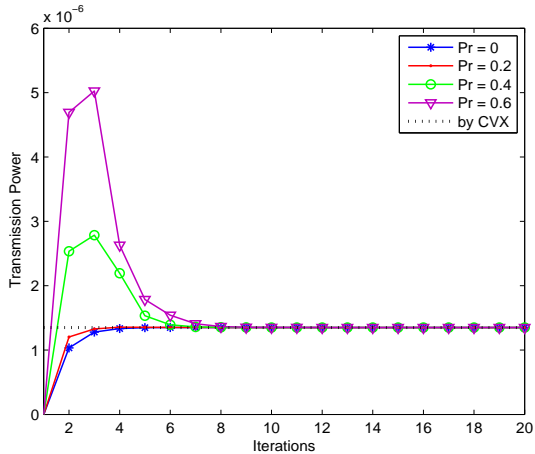}
	\caption{ Convergence of the asynchronous distributed beamforming and power control algorithm with different update failure probability. $M =4,~K=4$.}
	\label{fig:fig5}
\end{figure}

Next, we show the convergence of the asynchronous distributed beamforming and power control algorithm with $M=4,~K=4$. 
In \cite{bengtsson_optimal_2001}, the authors prove that the optimum of the SDP-relaxed master problem \eqref{master problem} is the same as the master problem as long as the master problem is feasible.
Based on this conclusion, we use CVX to compute the optimum of the SDP-relaxed master problem as a benchmark.
It is plotted as the black dot curve in Fig. \ref{fig:fig5}.
Furthermore, we assume $T_{\t{CSI}} \gg T_{\t{DC}} =T_{\t{BF}}=T_{\t{PC}}$. Specifically, the CSI is assumed to be fixed. In each iteration, the BSs update the beamforming vectors once and the users update the power allocations once. The dual variables are updated according to the probability $P$ in each iteration. 
Fig. \ref{fig:fig5} shows that the beamforming and power control algorithm converges to the benchmark under different $P$.
Besides, Fig. \ref{fig:fig5} shows that when $P$ is large, the total transmission power could have large fluctuations.
When $P=0$, the algorithm converges quickly.
This manifestation reminds us that in practice reliable backhaul networks are important for the performance of the algorithm.

\begin{figure}[t]
	\centering
		\includegraphics[width=0.45\textwidth]{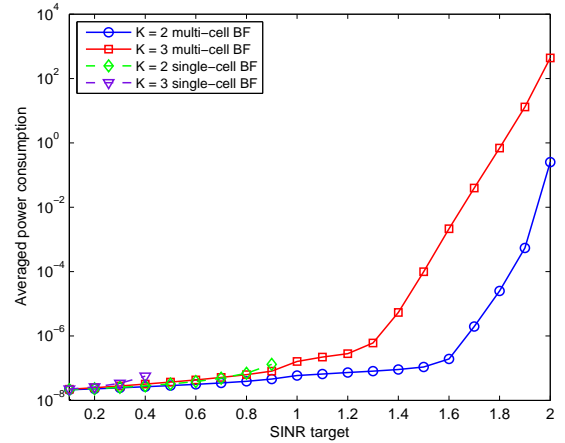}
	\caption{ Average power consumption versus different SINR targets. $M =4$}
	\label{fig:fig6}
\end{figure}

If a BS computes the beamforming vectors and power allocations treating the inter-cell interference as noise, we call this kind of beamforming the single-cell beamforming. The performance gain of multi-cell beamforming over single-cell beamforming is illustrated in Fig. \ref{fig:fig6}.
Each curve in Fig. \ref{fig:fig6} is plotted via $1000$-time averaging.
Only the setup, with which the master problem is feasible, is used.

In both the cases $M = 4,~K=2$ and $M=4,~K = 3$, the averaged power consumption of the multi-cell beamforming is less than that of the single-cell beamforming. 
Furthermore, the power consumption of the single-cell beamforming increases more quickly than the multi-cell beamforming. 
It is because the single-cell beamforming only coordinates intra-cell interference. When the SINR target increases, a BS needs to increase power to increase the SINR.
When all BSs do so, the inter-cell interference received by the users will increase. As a result, the BSs need to further increase power. However, when using multi-cell beamforming, inter-cell interference is coordinated. The BSs increase power mainly for increasing the users' SINRs.

Besides, the increase of the total power consumption implies that certain BS's power consumption must increase. In practice, the maximum transmission power of a BS is constrained.
Considering each BS's power constraint, the method used in \cite{yu_transmitter_2007,cai2012maxmin} can be used to find the minimum power consumption. Briefly, the power minimization problem can be decomposed into two levels by decomposing the BSs' power constraints. Our method here can be used to solve the lower level problem.

\section{Conclusion}
\label{sec:Conclusion}
In this paper, our main contribution is an asynchronous distributed beamforming and power control algorithm for the multi-cell networks. The algorithm is derived via investigating a multi-cell PMBP which assumes long-term CSI, non-reciprocal channel, and different noise variances at each user. 
Via exploring the special problem structure, we are able to propose a novel asynchronous iterative dual computation method to compute the dual of the multi-cell PMBP. 
Besides, beamforming and power control could be done within each cell by the BSs and the users asynchronously in a distributed fashion. At last, the convergence and performance of the algorithm are also demonstrated via simulation.

\appendices
\section{Proof of Proposition \ref{pro:feasibility}}
\label{proof of the feasibility pro}

We first introduce the definitions of Z-matrix, P-matrix, and K-matrix for the proof of this proposition.
\begin{defi}
(\cite[Definition 2]{jong_2010_design})
\romannumeral1)
	A matrix $\textbf{M} \in \mathbb{R}^{n\times n}$ is called Z-matrix if its off-diagonal entries are all non-positive;
\romannumeral2)
	A matrix $\textbf{M} \in \mathbb{R}^{n\times n}$ is called P-matrix if it reverses the sign of no nonzero vector, i.e., ${x}_{i}[\textbf{M}\boldsymbol{x}]_{i} \le 0 \Rightarrow \boldsymbol{x} = 0$;
\romannumeral3)
	If a matrix is both Z-matrix and P-matrix, it is called K-matrix.
\end{defi}

The K-matrix is related to the matrix radius according to the following lemma.
\begin{lm}
(\cite[Lemma 1]{jong_2010_design})
\label{K matrix and radius}
	Let $\textbf{M} \in \mathbb{R}^{n\times n}$ be a K-matrix, $\textbf{N} \in \mathbb{R}^{n\times n}$ be a non-negative matrix, then $\rho(\textbf{M}^{-1}\textbf{N})<1$ if and only if $\textbf{M} - \textbf{N}$ is a K-matrix.
\end{lm}

According to Lemma \ref{feasibility condition}, we need to prove that the inequalities \eqref{feasibility conditions 2} is the necessary and sufficient condition for $\rho( \mb{\Gamma}\mb{G} )<1$.
Since $\rho( \mb{\Gamma}\mb{G} ) = \rho( \mb{I}^{-1}\mb{\Gamma}\mb{G} )<1$ where $\mb{I}$ is a K-matrix, we conclude that $\mb{I} - \mb{\Gamma}\mb{G}$ is a K-matrix according to Lemma \ref{K matrix and radius}.
 
Furthermore, $\mb{I} - \mb{\Gamma}\mb{G}$ is a P-matrix whose components are 
\begin{align}
	&[\mb{I} - \mb{\Gamma}\mb{G}]_{(m,i),(n,j)} \notag \\
	= &\left\{
	\begin{array}{l l}
	- \frac{\gamma_{m,i}}{ \bs{w}_{m,i}^{H}\mb{R}_{m,m,i}\bs{w}_{m,i} / \bs{w}_{n,j}^{H}\mb{R}_{n,m,i}\bs{w}_{n,j} }, & (m,i) \neq (n,j) \\
	1, &  (m,i) = (n,j).
	\end{array}
	\right.
\end{align}
According to the definition of P-matrix, we have
\begin{align}
\label{P matrix}
	p_{m,i} [(\mb{I} -\mb{\Gamma}\mb{G})\bs{p}]_{m,i}
	>0,~\forall (m,i),
\end{align}
where $[(\mb{I} -\mb{\Gamma}\mb{G})\bs{p}]_{m,i}$ is the component of vector $(\mb{I} -\mb{\Gamma}\mb{G})\bs{p}$ corresponding to user $(m,i)$.

If inequalities \eqref{P matrix} hold for all $p_{m,i}>0$, we obtain 
$[(\mb{I} -\mb{\Gamma}\mb{G})\bs{p}]_{m,i} > 0,~\forall (m,i),$ which are equivalent to \eqref{feasibility conditions 2}.

\section{Proof of Lemma \ref{minimum non-negative eigenvalue}}
\label{proof of min-eig theorem}

To prove this lemma, we resort to the theoretical results about the \emph{semidefiniteness of a matrix pencil} \cite{golub1996matrix,kovavc1995trace,Liang20133085}.
We first introduce some necessary terminologies about matrix pencils, i.e., the following two definitions, which can be found in \cite[pp. 375-377]{golub1996matrix}.
The eigenvalues of a matrix pencil are defined as follows.
\begin{defi}(\emph{Eigenvalue and Eigenvector of a Matrix Pencil})
The set of all the form $\mb{A} - \mu\mb{B}$ with $\mu\in\mathbb{C}$ is said to be a \emph{pencil}. The eigenvalues of the pencil are elements of the set $\mu(\mb{A},\mb{B})$ defined by
$$
	\mu(\mb{A},\mb{B}) = \{ z \in \mathbb{C} : \text{det}( \mb{A} - z \mb{B} ) = 0\}.
	$$
If $\mu\in\mu(\mb{A},\mb{B})$ and 
$$
		\mb{A} \bs{x} = \mu\mb{B}\bs{x},~\bs{x}\neq 0,
$$
then $\bs{x}$ is referred to as an eigenvector of $\mb{A} - \mu\mb{B}$.
\end{defi}

\begin{defi}(\emph{Generalized Schur Decomposition})
\label{gener schur}
If $\mb{A}$ and $\mb{B}$ are in $\mathbb{C}^{n\times n}$, then there exist unitary $\mb{Q}$ and $\mb{Z}$ such that $\mb{Q}^{H}\mb{A}\mb{Z} = \mb{T}$ and $\mb{Q}^{H}\mb{B}\mb{Z} = \mb{S}$ are upper triangular. If for some $k$, $t_{kk}$ and $s_{kk}$ are both zero, then $\mu\in\mu(\mb{A},\mb{B}) = \mathbb{C}$. Otherwise 
$$
	\mu\in\mu(\mb{A},\mb{B}) = \{ \frac{t_{ii}}{s_{ii}} : s_{ii}\neq 0
	\}.
$$
\end{defi}

With the above definitions, we start the formal proof of Lemma \ref{minimum non-negative eigenvalue}.
More strict proof of similar results in the general case can be found in the study about semidefiniteness intervals of matrix pencil \cite[Corollary 3.7]{kovavc1995trace},\cite[Theorem 2.1]{Liang20133085}.
Our proof is given as follows.

\subsection{Proof of Existence and Finiteness of $\mu_{+}(\mb{A} , \mb{B}))$:}


According to \cite{stewart1979pertubation}, we have $$\mu_{+}(\mb{A},\mb{B}) = \underset{\bs{x}^{H}\mb{B}\bs{x} = 1}{\t{min}}~~\bs{x}^{H}\mb{A}\bs{x}.$$
Since $\mb{A}\succ 0,~\mb{B}\succeq 0$ by assumption, the minimum of $\underset{\bs{x}^{H}\mb{B}\bs{x} = 1}{\t{min}}~~\bs{x}^{H}\mb{A}\bs{x}$ exists and is finite due to convexity.

\subsection{Proof of $\mb{A} - \mu\mb{B} \succeq 0 \Rightarrow \mu \le \mu_{+}(\mb{A} , \mb{B})$:}


Since $\mb{A} - \mu\mb{B}$ is positive semidefinite, we have
\begin{align}
\label{xAx+xBx}
	\bs{x}^{H}\mb{A}\bs{x} - \mu\bs{x}^{H}\mb{B}\bs{x} \ge 0,~\forall \bs{x}\neq 0
\end{align}

Since $\mb{A} \succ 0,~\mb{B}\succeq 0$, we have $\bs{x}^{H}\mb{A}\bs{x} > 0$ and $\bs{x}^{H}\mb{B}\bs{x} \ge 0$.

If $\bs{x}\in\{ \bs{z}:\bs{z}^{H}\mb{B}\bs{z} = 0,~\bs{z} \neq 0 \}$, then $\mu \ge 0$ by assumption. 

If $\bs{x}\in\{ \bs{z}:\bs{z}^{H}\mb{B}\bs{z} > 0,~\bs{z} \neq 0 \}$, then $\mu \le \frac{\bs{x}^{H}\mb{A}\bs{x}}{\bs{x}^{H}\mb{B}\bs{x}}$ based on \eqref{xAx+xBx}.

Furthermore, $\mu \le \frac{\bs{x}^{H}\mb{A}\bs{x}}{\bs{x}^{H}\mb{B}\bs{x}}$ is equivalent to 
$$
\mu \le \underset{\bs{x}}{\t{min}}\left\{  \frac{\bs{x}^{H}\mb{A}\bs{x}}{\bs{x}^{H}\mb{B}\bs{x}} \right\}.
$$

The value of $\underset{\bs{x}}{\t{min}}\left\{\frac{\bs{x}^{H}\mb{A}\bs{x}}{\bs{x}^{H}\mb{B}\bs{x}}\right\}$ is equivalent to the optimal value of the following problem 
\begin{align}
\label{minimization problem for eigenvalue}
	\underset{\bs{x}}{\t{min}}~~\bs{x}^{H}\mb{A}\bs{x}~~~~\text{s.t.}~~\bs{x}^{H}\mb{B}\bs{x} = 1,
\end{align}
because we can always scale $\bs{x}$ such that $\bs{x}^{H}\mb{B}\bs{x} = 1$.
According to \cite{stewart1979pertubation}, the minimizing vector is the eigenvector of matrix pencil $\mb{A} - \mu\mb{B}$, whereas the minimum is the corresponding eigenvalue. Thus, we can conclude that $\mu\le \mu_{+}(\mb{A},\mb{B})$, i.e., being smaller or equal to the minimum non-negative eigenvalue of matrix pencil $\mb{A} - \mu\mb{B}$.

\subsection{Proof of $\mb{A} - \mu\mb{B} \succeq 0 \Leftarrow \mu \le \mu_{+}(\mb{A} , \mb{B})$:}


This proof is based on the determinant of the matrix pencil $\mb{A} - \mu\mb{B}$. 

Let $\beta_{1}(\mu)\ge\beta_{2}(\mu)\ge,...,\ge\beta_{N}(\mu)$ be the eigenvalues of matrix $\mb{A} - \mu \mb{B}$ given $\mu$. We have 
\begin{align}
\label{det eigen}
	\text{det}( \mb{A} - \mu \mb{B} ) 
	= \prod_{k=1,...,N}\beta_{k}(\mu).
\end{align}
In addition, we have the following lemma.
\begin{lm}
	\label{continuity of eigenvalue}
	Eigenvalues $\beta_{k}(\mu),~k=1,...,N$, are continuous functions of $\mu$.
\end{lm}
\begin{proof}
	Let $\mu_{0}\in (0,\mu_{+}(\mathbf{A},\mathbf{B}))$. To prove the continuity is to prove
$$
\underset{\epsilon \to 0}{\lim} \beta_{k}(\mu_{0}+\epsilon) = \underset{\mu \to \mu_{0}}{\lim}\beta_{k}(\mu) = \beta_{k}(\mu_{0}),~\forall k.
$$
Let $\lambda_{k}(\mathbf{A})$ represents the $k$th eigenvalue of matrix $\mathbf{A}$, assuming $\lambda_{k}(\mathbf{A}) \ge \lambda_{k+1}(\mathbf{A})$. We have 
$$
\beta_{k}(\mu_{0}+\epsilon) = \lambda_{k}\left( \mathbf{A} - (\mu_{0}+\epsilon) \mathbf{B}   \right) = \lambda_{k}\left( \mathbf{A} - \mu_{0}\mathbf{B} - \epsilon\mathbf{B} \right),~\forall k.
$$ 
The matrix $\epsilon\mathbf{B}$ can be viewed as a perturbation to matrix $\mathbf{A} - \mu_{0}\mathbf{B}$. 

Next, we resort to the \emph{matrix perturbation theory}. Since $\mathbf{A},\mathbf{B}$ are symmetric matrices, according to the Hoffman-Wielandt Theorem \cite[pp. 189]{stewart1990matrix}, we have
$$
\sum_{k} \left( \lambda_{k}\left( \mathbf{A} - \mu_{0}\mathbf{B} - \epsilon\mathbf{B}\right)  - \lambda_{k}(\mathbf{A} - \mu_{0}\mathbf{B})  
  \right)^{2} \le \epsilon^{2}|| \mathbf{B}||_{\text{F}}^{2},
$$
where $||\cdot||_{\text{F}}$ represents Frobenius matrix norm.
We can see that 
$$\left( \lambda_{k}\left( \mathbf{A} - \mu_{0}\mathbf{B} - \epsilon\mathbf{B}\right)  - \lambda_{k}(\mathbf{A} - \mu_{0}\mathbf{B})  
  \right)^{2} \ge 0,~\forall k.$$
In addition, since $|| \mathbf{B}||_{\text{F}}^{2}$ is finite, $\underset{\epsilon \to 0}{\lim}\epsilon^{2}|| \mathbf{B}||_{\text{F}}^{2} = 0$.
Therefore, we can conclude that 
$$\underset{\epsilon \to 0}{\lim} \left( \lambda_{k}\left( \mathbf{A} - \mu_{0}\mathbf{B} - \epsilon\mathbf{B}\right)  - \lambda_{k}(\mathbf{A} - \mu_{0}\mathbf{B})  
  \right)^{2} = 0,~\forall k.  $$
That is, $\underset{\epsilon \to 0}{\lim} \beta_{k}(\mu_{0}+\epsilon) =  \beta_{k}(\mu_{0}),~\forall k.$
This completes the proof. 
\end{proof}

Moreover, based on Definition \ref{gener schur}, we have 
\begin{align}
	\text{det}( \mb{A} - \mu \mb{B} ) &= \text{det}( \mb{T} - \mu \mb{S} )\text{det}(\mb{Q})\text{det}(\mb{Z}) \notag \\
	&= \prod_{k = 1,...,N} (t_{kk} - \mu s_{kk}). \notag 
\end{align}
With \eqref{det eigen}, we have 
\begin{align}
\label{det eigen and schur}
	\text{det}( \mb{A} - \mu \mb{B} ) = \prod_{k = 1,...,N} (t_{kk} - \mu s_{kk}) 
	= \prod_{k=1,...,N}\beta_{k}(\mu).
\end{align}

It is easy to see that if $\mu = 0$, $\text{det}( \mb{A} - \mu \mb{B} ) > 0$ and all the eigenvalues are positive.
Let $\mu$ increase continuously from $0$. From the polynomial in \eqref{det eigen and schur}, we can see that the determinant keeps to be positive until $\mu = \mu_{+}(\mb{A},\mb{B})$ where $\text{det}( \mb{A} - \mu \mb{B} ) = 0$.
With Lemma \ref{continuity of eigenvalue}, it implies that no $\beta_{k}(\mu)$ changes sign during the process. So, we can conclude that if $0\le \mu \le \mu_{+}(\mb{A} , \mb{B})$, then $\mb{A} - \mu\mb{B} \succeq 0$.

\section{Proof of Theorem \ref{dual convergence}}
\label{proof of the dual convergence}
The proof of this theorem is motivated by the work of \cite{yates1995framework}.
We need to prove $\bs{J}(\bs{\lambda})$ is a standard function as defined in \cite{yates1995framework}. That is, for all
$\bs{\lambda}\ge 0$, the following properties are satisfied:
\begin{itemize}
	\item Positivity: If $\bs{\lambda}\ge 0$, $\bs{J}(\bs{\lambda}) \ge 0$,
	\item Monotonicity: If $\bs{\lambda} \ge \bs{\lambda}^{\prime} $, $\bs{J}(\bs{\lambda})\ge \bs{J}(\bs{\lambda}^{\prime})$,
	\item Scalability: $\alpha \bs{J}(\bs{\lambda})>\bs{J}(\alpha\bs{\lambda}),~\forall \alpha >1$.
\end{itemize}
If $\bs{J}(\bs{\lambda})$ is a standard function, then the remaining proof is identical to the proof in \cite{yates1995framework}. Our contribution here is to prove that the specific function $\bs{J}(\bs{\lambda})$ is standard.

\subsection{Proof of Positivity:} 

The dual variables $\bs{\lambda}$ are greater than or equal to zero trivially.
Since $J_{m,i}(\bs{\lambda}),~\forall (m,i)$, are the non-negative eigenvalues, we have $\bs{J}(\bs{\lambda}) \ge 0$.

\subsection{Proof of Monotonicity:} 

To prove the monotonicity, we first propose the following lemma,
\begin{lm}
\label{lm:monotonicity}

Let $N\times N$ Hermitian matrices $\mb{B},\mb{C},\mb{D} \succeq 0$. If $\beta^{\prime}\ge \beta \ge 0$, then 
\begin{align}
	\label{eq:monotonicity}
	\mu_{+}(\mb{C} + \beta^{\prime}\mb{D},\mb{B}) \ge \mu_{+}(\mb{C} + \beta\mb{D},\mb{B}),
\end{align}

\end{lm}
\begin{proof}
Let $\mu_{+} = \mu_{+}(\mb{C} + \beta\mb{D},\mb{B})$ and $\bs{v}$ be the corresponding generalized eigenvector.
Similarly, let $\mu_{+}^{\prime} = \mu_{+}(\mb{C} + \beta^{\prime}\mb{D},\mb{B})$ and $\bs{u}$ be the corresponding generalized eigenvector.

From equation \eqref{minimization problem for eigenvalue}, we have
$$
\mu_{+} = \bs{v}^{H}(\mb{C} + \beta\mb{D})\bs{v},
$$
while 
$$	\mu_{+}^{\prime} = \bs{u}^{H}(\mb{C} + \beta^{\prime}\mb{D})\bs{u} 
	= \bs{u}^{H}(\mb{C} + \beta\mb{D})\bs{u} + (\beta^{\prime} - \beta)\bs{u}^{H}\mb{D}\bs{u}.$$

Since $\bs{v}$ minimize $\bs{x}^{H}(\mb{C} + \beta\mb{D})\bs{x}$ subject to $\bs{x}^{H}\mb{B}\bs{x} = 1$, we have
$$
	\bs{u}^{H}(\mb{C} + \beta\mb{D})\bs{u} \ge \bs{v}^{H}(\mb{C} + \beta\mb{D})\bs{v}.
$$

Furthermore, since $\beta^{\prime}\ge \beta$ and $\mb{D}\succeq 0$, we have 
$$(\beta^{\prime} - \beta)\bs{u}^{H}\mb{D}\bs{u} \ge 0.$$
Then, we can conclude that $\mu_{+}^{\prime} \ge \mu_{+}$.
\end{proof}

Suppose that, from $\bs{\lambda}$ to $\bs{\lambda}^{\prime}$, only $\lambda_{n,j}$ changes to $\lambda_{n,j}^{\prime}$ while the other entries remain unchanged.

Let
\begin{align}
	\mb{B} &= \mb{R}_{m,m,i}, \notag \\
	\mb{C} &= \left(1+\frac{1}{\gamma_{m,i}}\right)^{-1}\left( \mb{I}+ \sum_{(l,i)\neq(n,j)}\lambda_{l,i}\mb{R}_{m,l,i}\right), \notag \\
	\mb{D} &= \left(1+\frac{1}{\gamma_{m,i}}\right)^{-1}\mb{R}_{m,n,j}.
\end{align}

According to Lemma \ref{lm:monotonicity}, if $\lambda_{n,j}^{\prime} \ge \lambda_{n,j}$, then
$$
\mu_{+}(\mb{C} + \lambda_{n,j}^{\prime}\mb{D},\mb{B}) \ge \mu_{+}(\mb{C} + \lambda_{n,j}\mb{D},\mb{B}),
$$
which means $J_{m,i}(\bs{\lambda}^{\prime}) \ge J_{m,i}(\bs{\lambda}),~\forall (m,i)$.

If multiple entries increase from $\bs{\lambda}$ to $\bs{\lambda}^{\prime}$, it is equivalent to let them increase one by one.
Consequently, if $\bs{\lambda} \ge \bs{\lambda}^{\prime} $, $\bs{J}(\bs{\lambda})\ge \bs{J}(\bs{\lambda}^{\prime})$ according to Lemma \ref{lm:monotonicity}.

\subsection{Proof of Scalability:}

Below, we prove that $\alpha J_{m,i}(\bs{\lambda}) > J_{m,i}(\alpha \bs{\lambda}),~\forall (m,i),$ with $\alpha>1$.

Let 
\begin{align}
	\mb{B} &= \mb{R}_{m,m,i}, \notag \\
	\mb{C} &= \left(1+\frac{1}{\gamma_{m,i}}\right)^{-1}\sum_{n,j}\alpha\lambda_{n,j}\mb{R}_{m,n,j}, \notag \\
	\mb{D} &= \left(1+\frac{1}{\gamma_{m,i}}\right)^{-1}\mb{I}.
\end{align}

Then, we have
$$
J_{m,i}(\alpha \bs{\lambda}) = \mu_{+}(\mb{C} + \mb{D},\mb{B}),
$$
and 
$$
 \alpha J_{m,i}(\bs{\lambda}) = \mu_{+}(\mb{C} + \alpha\mb{D},\mb{B})
$$
becuase of 
\begin{align}
		&\mu_{+}(\mathbf{C}+\alpha \mathbf{D},\mathbf{B}) =\notag \\
		&
	\underset{\boldsymbol{x}^{H}\mathbf{R}_{m,m,i}\boldsymbol{x} = 1}{\text{min}}
		\boldsymbol{x}^{H} \Big(1+\frac{1}{\gamma_{m,i}}\Big)^{-1}
		\Big(
		\alpha \mathbf{I}+ \sum_{n,j}\alpha \lambda_{n,j}\mathbf{R}_{m,n,j}
		\Big) \boldsymbol{x} \notag \\
		&= \alpha J_{m,i}(\bs{\lambda}). \notag 
\end{align}

Recall that, in the proof of Lemma \ref{lm:monotonicity}, if $\mb{D} \succ 0$, $(\beta^{\prime} - \beta)\bs{u}^{H}\mb{D}\bs{u} > 0$, then
$$
\mu_{+}(\mb{C} + \beta^{\prime}\mb{D},\mb{B}) > \mu_{+}(\mb{C} + \beta\mb{D},\mb{B}).$$ 
Since $\alpha > 1$, according to Lemma \ref{lm:monotonicity}, we can conclude that $\alpha \bs{J}(\bs{\lambda}) > \bs{J}(\alpha\bs{\lambda}),~\forall \alpha >1$.

With the above results, the theorem can be proved by a line of reasoning similar to that in \cite{yates1995framework}.
The sketch of the remaining proof is as follows.
First prove that the uniqueness of the fixed point.
If $\bs{\lambda}(0) \ge \bs{J}(\bs{\lambda}(0)) $, the sequences are decreasing.
If $\bs{\lambda}(0) \le \bs{J}(\bs{\lambda}(0))  $, the sequences are increasing.
The sequences converge to the unique fixed point.

\bibliographystyle{IEEEtran}
\bibliography{beamforming}

\begin{thebibliography}{10}
\providecommand{\url}[1]{#1}
\csname url@samestyle\endcsname
\providecommand{\newblock}{\relax}
\providecommand{\bibinfo}[2]{#2}
\providecommand{\BIBentrySTDinterwordspacing}{\spaceskip=0pt\relax}
\providecommand{\BIBentryALTinterwordstretchfactor}{4}
\providecommand{\BIBentryALTinterwordspacing}{\spaceskip=\fontdimen2\font plus
\BIBentryALTinterwordstretchfactor\fontdimen3\font minus
  \fontdimen4\font\relax}
\providecommand{\BIBforeignlanguage}[2]{{%
\expandafter\ifx\csname l@#1\endcsname\relax
\typeout{** WARNING: IEEEtran.bst: No hyphenation pattern has been}%
\typeout{** loaded for the language `#1'. Using the pattern for}%
\typeout{** the default language instead.}%
\else
\language=\csname l@#1\endcsname
\fi
#2}}
\providecommand{\BIBdecl}{\relax}
\BIBdecl

\bibitem{rashid-farrokhi_transmit_1998}
F.~Rashid-Farrokhi, K.~Liu, and L.~Tassiulas, ``Transmit beamforming and power
  control for cellular wireless systems,'' \emph{IEEE J. Sel. Areas Commun.},
  vol.~16, no.~8, pp. 1437-- 1450, Oct. 1998.

\bibitem{visotsky_optimum_1999}
E.~Visotsky and U.~Madhow, ``Optimum beamforming using transmit antenna
  arrays,'' in \emph{Proc. Veh. Technol. Conf. (VTC)}, Houston, TX, 1999, pp.
  851--856.

\bibitem{bengtsson_optimal_1999}
M.~Bengtsson and B.~Ottersten, ``Optimal downlink beamforming using
  semidefinite optimization,'' in \emph{Proc. 37th Annu. Allerton Conf.
  Commun., Control and Computing}, 1999, pp. 987--996.

\bibitem{schubert_solution_2004}
M.~Schubert and H.~Boche, ``Solution of the multiuser downlink beamforming
  problem with individual {SINR} constraints,'' \emph{{IEEE} Trans. Veh.
  Technol.}, vol.~53, no.~1, pp. 18--28, Jan. 2004.

\bibitem{wiesel_linear_2006}
A.~Wiesel, Y.~Eldar, and S.~Shamai, ``Linear precoding via conic optimization
  for fixed {MIMO} receivers,'' \emph{{IEEE} Trans. Signal Process.}, vol.~54,
  no.~1, pp. 161--176, Jan. 2006.

\bibitem{yu_transmitter_2007}
W.~Yu and T.~Lan, ``Transmitter optimization for the multi-antenna downlink
  with per-antenna power constraints,'' \emph{{IEEE} Trans. Signal Process.},
  vol.~55, no.~6, pp. 2646--2660, Jun. 2007.

\bibitem{bengtsson_optimal_2001}
M.~Bengtsson and B.~Ottersten, ``Optimal and suboptimal transmit beamforming,''
  in \emph{Handbook of Antennas in Wireless Communications}, 1st~ed.\hskip 1em
  plus 0.5em minus 0.4em\relax CRC Press, 2001, pp. 18--1--18--33.

\bibitem{lin2012distributed}
X.~Lin and T.~M. Lok, ``Distributed power control for one-to-many transmissions
  in gaussian interference channels,'' \emph{{IEEE} Trans. Commun.}, vol.~60,
  no.~8, pp. 2363--2375, Aug. 2012.

\bibitem{bertsekas1989parallel}
D.~P. Bertsekas and J.~N. Tsitsiklis, \emph{Parallel and distributed
  computation}.\hskip 1em plus 0.5em minus 0.4em\relax Prentice Hall Inc.,
  1989.

\bibitem{stewart1979pertubation}
G.~Stewart, ``Pertubation bounds for the definite generalized eigenvalue
  problem,'' \emph{Linear algebra and its applications}, vol.~23, pp. 69--85,
  Feb. 1979.

\bibitem{kovavc1995trace}
J.~Kova{\v{c}}-Striko and K.~Veseli{\'c}, ``Trace minimization and definiteness
  of symmetric pencils,'' \emph{Linear algebra and its applications}, vol. 216,
  pp. 139--158, Feb. 1995.

\bibitem{Liang20133085}
X.~Liang, R.-C. Li, and Z.~Bai, ``Trace minimization principles for positive
  semi-definite pencils,'' \emph{Linear Algebra and its Applications}, vol.
  438, no.~7, pp. 3085--3106, Apr. 2013.

\bibitem{yates1995framework}
R.~D. Yates, ``A framework for uplink power control in cellular radio
  systems,'' \emph{{IEEE} J. Sel. Areas Commun.}, vol.~13, no.~7, pp.
  1341--1347, Sep. 1995.

\bibitem{botella2008coordination}
C.~Botella, G.~Pinero, A.~Gonzalez, and M.~De~Diego, ``Coordination in a
  multi-cell multi-antenna multi-user w-cdma system: A beamforming approach,''
  \emph{{IEEE} Trans. Wireless Commun.}, vol.~7, no.~11, pp. 4479--4485, Nov.
  2008.

\bibitem{huang_distributed_2011}
Y.~Huang, G.~Zheng, M.~Bengtsson, K.-K. Wong, L.~Yang, and B.~Ottersten,
  ``Distributed multicell beamforming with limited intercell coordination,''
  \emph{{IEEE} Trans. Signal Process.}, vol.~59, no.~2, pp. 728--738, Feb.
  2011.

\bibitem{dahrouj_coordinated_2010}
H.~Dahrouj and W.~Yu, ``Coordinated beamforming for the multicell multi-antenna
  wireless system,'' \emph{{IEEE} Trans. Wireless Commun.}, vol.~9, no.~5, pp.
  1748--1759, May 2010.

\bibitem{gershman_convex_2010}
A.~Gershman, N.~Sidiropoulos, S.~Shahbazpanahi, M.~Bengtsson, and B.~Ottersten,
  ``Convex optimization-based beamforming,'' \emph{{IEEE} Signal Process.
  Mag.}, vol.~27, no.~3, pp. 62--75, May 2010.

\bibitem{vucic_robust_2009}
N.~Vucic and H.~Boche, ``Robust {QoS-Constrained} optimization of downlink
  multiuser {MISO} systems,'' \emph{{IEEE} Trans. Signal Process.}, vol.~57,
  no.~2, pp. 714--725, Feb. 2009.

\bibitem{huang_dual_2010}
Y.~Huang and D.~Palomar, ``A dual perspective on separable semidefinite
  programming with applications to optimal downlink beamforming,'' \emph{{IEEE}
  Trans. Signal Process.}, vol.~58, no.~8, pp. 4254--4271, Aug. 2010.

\bibitem{dartmann_duality_2013}
G.~Dartmann, X.~Gong, W.~Afzal, and G.~Ascheid, ``On the duality of the max-min
  beamforming problem with per-antenna and per-antenna-array power
  constraints,'' \emph{{IEEE} Trans. Veh. Technol.}, vol.~62, no.~2, pp.
  606--619, Feb. 2013.

\bibitem{Xiang2013coordinated}
Z.~Xiang, M.~Tao, and X.~Wang, ``Coordinated multicast beamforming in multicell
  networks,'' \emph{{IEEE} Trans. Wireless Commun.}, vol.~12, no.~1, pp.
  12--21, Jan. 2013.

\bibitem{cai2012maxmin}
D.~Cai, T.~Quek, C.-W. Tan, and S.~Low, ``Max-min sinr coordinated multipoint
  downlink transmission--duality and algorithms,'' \emph{IEEE Trans. Signal
  Process.}, vol.~60, no.~10, pp. 5384--5395, Oct. 2012.

\bibitem{jeong2011beamforming}
Y.~Jeong, T.~Q. Quek, and H.~Shin, ``Beamforming optimization for multiuser
  two-tier networks,'' \emph{IEEE J. Commun. Networks}, vol.~13, no.~4, pp.
  327--338, Aug. 2011.

\bibitem{noh_beamforming_2013}
J.-H. Noh and S.-J. Oh, ``Beamforming in a multi-user cognitive radio system
  with partial channel state information,'' \emph{{IEEE} Trans. Wireless
  Commun.}, vol.~12, no.~2, pp. 616--625, Feb. 2013.

\bibitem{zhang_cooperative_2010}
R.~Zhang and S.~Cui, ``Cooperative interference management with {MISO}
  beamforming,'' \emph{{IEEE} Trans. Signal Process.}, vol.~58, no.~10, pp.
  5450--5458, Oct. 2010.

\bibitem{lim_energy-efficient_2013}
G.~Lim and J.~Cimini, {L.J.}, ``Energy-efficient cooperative beamforming in
  clustered wireless networks,'' \emph{{IEEE} Trans. Wireless Commun.},
  vol.~12, no.~3, pp. 1376--1385, Mar. 2013.

\bibitem{wang_robust_2013}
K.-Y. Wang, N.~Jacklin, Z.~Ding, and C.-Y. Chi, ``Robust {MISO} transmit
  optimization under outage-based {QoS} constraints in two-tier heterogeneous
  networks,'' \emph{{IEEE} Trans. Wireless Commun.}, vol.~12, no.~4, pp.
  1883--1897, 2013.

\bibitem{shendesign}
S.~Shen, H.~Fang, and T.~M. Lok, ``Design of downlink beamformer for real-time
  and non-real-time services,'' in \emph{Proc. {IEEE} Wireless Commun. and
  Networking Conf. (WCNC)}, Shanghai, 2013, pp. 3529--3534.

\bibitem{horn1990matrix}
R.~A. Horn and C.~R. Johnson, \emph{Matrix analysis}.\hskip 1em plus 0.5em
  minus 0.4em\relax Cambridge Univ. press, 1990.

\bibitem{jong_2010_design}
J.-S. Pang, G.~Scutari, D.~Palomar, and F.~Facchinei, ``Design of cognitive
  radio systems under temperature-interference constraints: A variational
  inequality approach,'' \emph{{IEEE} Trans. Signal Process.}, vol.~58, no.~6,
  pp. 3251--3271, Jun. 2010.

\bibitem{somekh_cth11-2_2006}
O.~Somekh, O.~Simeone, Y.~Bar-Ness, and A.~Haimovich, ``{CTH11-2}: Distributed
  multi-cell zero-forcing beamforming in cellular downlink channels,'' in
  \emph{Proc. Global Telecommun. Conf. ({GLOBECOM})}, San Francisco, CA, 2006,
  pp. 1--6.

\bibitem{vemula2006inter}
M.~Vemula, D.~Avidor, J.~Ling, and C.~Papadias, ``Inter-cell coordination,
  opportunistic beamforming and scheduling,'' in \emph{Proc. {IEEE} Inter.
  Conf. Communications ({ICC})}, Istanbul, 2006, pp. 5319--5324.

\bibitem{stridh_system_2006}
R.~Stridh, M.~Bengtsson, and B.~Ottersten, ``System evaluation of optimal
  downlink beamforming with congestion control in wireless communication,''
  \emph{{IEEE} Trans. Wireless Commun.}, vol.~5, no.~4, pp. 743--751, Apr.
  2006.

\bibitem{grant2008cvx}
M.~Grant, S.~Boyd, and Y.~Ye, ``{CVX}: Matlab software for disciplined convex
  programming,'' 2008.

\bibitem{golub1996matrix}
G.~H. Golub and C.~F. Van~Loan, \emph{Matrix computations}.\hskip 1em plus
  0.5em minus 0.4em\relax Johns Hopkins Univ. Press, 1996.

\bibitem{cai2011unified}
D.~Cai, T.~Quek, and C.-W. Tan, ``A unified analysis of max-min weighted sinr
  for mimo downlink system,'' \emph{IEEE Trans. Signal Process.}, vol.~59,
  no.~8, pp. 3850--3862, Aug. 2011.

\bibitem{stewart1990matrix}
G.~W. Stewart and J.-g. Sun, \emph{Matrix perturbation theory}.\hskip 1em plus
  0.5em minus 0.4em\relax Academic press, 1990.

\end{thebibliography}

\end{document}